\documentclass[11pt]{article}

\usepackage[margin=1in]{geometry}
\usepackage{amsmath,amssymb,amsfonts,amsthm,mathtools}
\usepackage{booktabs,longtable,tabularx,array,multirow}
\usepackage{graphicx}
\usepackage{enumitem}
\usepackage{xcolor}
\usepackage{microtype}
\usepackage[numbers,sort&compress]{natbib}
\usepackage{xurl}
\usepackage{hyperref}
\usepackage{float}
\usepackage{authblk}
\hypersetup{colorlinks=true,linkcolor=blue,citecolor=blue,urlcolor=blue,hypertexnames=false}

\newtheorem{proposition}{Proposition}

\newcommand{\R}{\mathbb{R}}
\newcommand{\C}{\mathbb{C}}
\newcommand{\CP}{\mathbb{CP}}
\newcommand{\Gr}{\operatorname{Gr}}

\newcommand{\tr}{\operatorname{tr}}

\newcommand{\rank}{\operatorname{rank}}

\newcommand{\ket}[1]{|#1\rangle}
\newcommand{\bra}[1]{\langle #1|}
\newcommand{\braket}[2]{\langle #1|#2\rangle}

\title{Invariance Audits for Quantum Kernels and Variational Rewinding:\\
A Real-to-Hermitian Taxonomy of Projector, Flag, Anchor, and Density Geometry}
\author[1,2]{Azadeh Alavi\thanks{All authors contributed equally. All authors are corresponding authors. Correspondence: \href{mailto:azadeh.alavi@rmit.edu.au}{azadeh.alavi@rmit.edu.au}; \href{mailto:fatima@pr2aid.com}{fatima@pr2aid.com}; \href{mailto:znj@pr2aid.com}{znj@pr2aid.com}.}}
\author[2]{Fatemeh Kouchmeshki}
\author[2]{Hossein Akhoundi}
\affil[1]{RMIT University, Melbourne, VIC, Australia}
\affil[2]{Pattern Recognition Pty Ltd, Australia}
\date{8 July 2026}

\begin{document}
\maketitle

\begin{abstract}
Machine-learning models often replace vectors by normalized directions, projectors, covariances, subspaces, ordered flags, quantum states, or density operators before any classifier is fitted. This replacement is an invariance decision: it determines which distinctions are kept and which are quotiented out. We develop a self-contained real-to-Hermitian taxonomy for auditing such representations in quantum machine learning. On the real side, we formalize Grassmann and flag projector kernels, prove positive semidefiniteness and block-gauge invariance of a weighted flag kernel, and give a same-span block-swap witness showing when whole-span Grassmann geometry must fail while ordered flags succeed. On the quantum side, we prove that a noiseless fidelity kernel is exactly a Hilbert--Schmidt inner product of rank-one Hermitian projectors, $K(x,y)=|\langle\phi(x)|\phi(y)\rangle|^2=\tr(P_xP_y)$, and that a QVR-style return probability is exactly an anchor-overlap score $p_\theta(x)=\tr(P_xA_\theta)$. Rank-$r$ returns are complex Grassmann anchors, while mixed or multimodal class models are naturally density or positive-semidefinite anchors. Controlled vector, subspace, statevector, anomaly, finite-shot, and quotient-witness experiments support the same conclusion: quantum and geometric lifts are useful when their invariances match the task, and fail correctly when discarded information is label-bearing. The paper makes no hardware-speedup or quantum-advantage claim.
\end{abstract}

\noindent\textbf{MSC/AMS subject classifications.} 53C30, 62H30, 68T07, 81P68, 15A18.

\section{Introduction}
A machine-learning sample is often written as a vector $x\in\R^d$, but many models do not use $x$ directly. They first replace it by another object:
\begin{equation}
 x\longmapsto \frac{x}{\|x\|},\qquad
 x\longmapsto xx^T,\qquad
 x\longmapsto C_x,\qquad
 x\longmapsto X_xX_x^T,\qquad
 x\longmapsto \ket{\psi(x)}\bra{\psi(x)}.
\end{equation}
Each replacement is a scientific decision. Unit normalization removes radius. A real rank-one projector removes sign. A Hermitian pure-state projector removes global phase. A covariance matrix emphasizes local variability and correlation. A Grassmann point keeps a subspace but removes the basis used to represent it. A flag keeps ordered block information. A quantum circuit maps the sample into Hilbert space and then scores overlaps, measurements, or return events.

The thesis of this paper is that these replacements should be audited as \emph{invariance choices}. We should not ask only whether a method is Euclidean, geometric, or quantum. We should ask what information the lift keeps, what it erases, and whether the erased information is truly nuisance for the task.

This manuscript focuses on the real-to-Hermitian geometry behind these choices: real projectors, Grassmann and flag projectors, Hermitian projectors, complex projective geometry, quantum fidelity kernels, and QVR-style return scores. It is related to the observation that supervised quantum machine-learning models can be understood as kernel methods \cite{schuld_kernel_methods}; our contribution is complementary. We give the exact projector, flag, anchor, density, and decision-operator dictionary, together with invariance witnesses that identify when a chosen quotient is appropriate. The scope is deliberately limited to fixed encoders, kernels, anchor returns, finite-shot approximations, and density or subspace extensions. Trainable circuit-design architectures and implementation details not necessary for the mathematical results are outside the scope of this paper.

\subsection{Why QVR requires a precise geometric statement}
Quantum Variational Rewinding (QVR) was introduced as a hybrid quantum-classical method for time-series anomaly detection: data are encoded into quantum states, a parameterized unitary is trained so normal data return close to a target event, and new samples are scored by a return or reconstruction statistic \cite{baker_qvr_2022,pennylane_qvr}. This paper does not claim to reproduce that anomaly-detection benchmark. Instead, it asks a more basic question: once the encoder, unitary, and success measurement are fixed, what geometric object is the return score computing?

The answer is exact. If $P_x=\ket{\phi(x)}\bra{\phi(x)}$ and $A_\theta=U_\theta^*\Pi U_\theta$, then
\begin{equation}
 \|\Pi U_\theta\ket{\phi(x)}\|_2^2=\tr(P_xA_\theta).
\end{equation}
Thus the ideal QVR return probability already lives in Hermitian projector geometry. The difference between QVR, a quantum-kernel SVM, a density-center model, and a rank-$r$ subspace-return model is not the presence or absence of ``quantumness'' at this algebraic level. It is the choice of encoder, anchor family, observable, classifier, finite-shot/noise model, and training objective.

\subsection{Contributions}
The paper contributes the following.
\begin{enumerate}[leftmargin=2em]
\item We give a self-contained taxonomy and dictionary from real projective, covariance, Grassmann, flag, complex projective, anchor, and density-matrix lifts to the invariances they impose.
\item We prove the real geometric guardrails needed by the taxonomy: the weighted flag projection kernel is positive semidefinite and block-gauge invariant; whole-span Grassmann features cannot separate same-span block swaps; and monotone spectral activation followed by top-$p$ reprojection is identity on exact projectors.
\item We prove that a noiseless quantum fidelity kernel is exactly a Hermitian projector kernel and therefore a positive semidefinite kernel on the embedded projective state manifold.
\item We prove that a QVR-style return probability is exactly an anchor-overlap score, with rank-one, rank-$r$, density, and general Hermitian-observable variants corresponding to distinct geometric models.
\item We separate three objects that are easily conflated: pairwise quantum-kernel SVMs, single-anchor QVR-style return models, and density/subspace class-center models.
\item We report controlled real-geometry, statevector, finite-shot, QVR-style anchor, density-center, anomaly, block-swap, and phase/radius witness experiments. We treat them as invariance witnesses rather than operational benchmarks.
\item We state conservative limitations that distinguish the paper's mathematical and diagnostic claims from hardware-speedup, applied-system, or quantum-advantage claims.
\end{enumerate}

\section{A taxonomy of quantum geometric representations}
The identities in this paper can be read as a descriptive taxonomy of quantum geometric representations. The taxonomy is mathematical rather than procedural: it classifies the object being scored, the quotient or invariance imposed by that object, and the type of score produced. It is not a runtime execution policy or commercial software architecture.

\begin{table}[H]
\centering
\caption{Taxonomy of representation objects used in the audit. The table is descriptive: it states the geometry and imposed invariance of each representation, not an operational workflow or deployment rule.}
\label{tab:taxonomy}
\scriptsize
\renewcommand{\arraystretch}{1.15}
\begin{tabularx}{\textwidth}{>{\raggedright\arraybackslash}p{0.13\textwidth}
>{\raggedright\arraybackslash}p{0.14\textwidth}
>{\raggedright\arraybackslash}p{0.21\textwidth}
>{\raggedright\arraybackslash}p{0.17\textwidth}
>{\raggedright\arraybackslash}X}
\toprule
Family & Object & Main invariance / quotient & Typical score & Main audit question\\
\midrule
Vector & $x\in\R^d$ & none beyond preprocessing & distance, dot product, RBF & Are scale, units, and preprocessing fitted safely?\\
Unit direction & $x/\|x\|$ & radius removed & cosine similarity & Is magnitude nuisance or signal?\\
Real projector & $uu^T$ & sign and radius removed & $\tr(P_uP_v)$ & Is line information sufficient?\\
Magnitude PSD & $xx^T$ & sign removed, scale retained & $(x^Ty)^2$ & Should magnitude be preserved?\\
Grassmann & $XX^T$ & basis removed inside a subspace & $\|X^TY\|_F^2$ & Is whole-span information enough?\\
Flag & $(X_bX_b^T)_b$ & basis removed inside ordered blocks & block-weighted trace sum & Does block order carry signal?\\
Pure quantum state & $\ket{\psi}$ & normalized Hilbert-state representation & amplitudes, measurements & What encoder-induced information is present?\\
Hermitian projector & $\ket{\psi}\bra{\psi}$ & global phase removed & $\tr(P_xP_y)$ & Is projective phase quotient appropriate?\\
Rank-$r$ anchor & $A\in\Gr_{\C}(r,D)$ & basis inside success subspace removed & $\tr(P_xA)$ & Does the target occupy a subspace?\\
Density anchor & $\rho\succeq0$, $\tr(\rho)=1$ & pure state replaced by mixture & $\tr(P_x\rho)$ or fidelity & Is the class/noise model mixed or multimodal?\\
Hermitian decision operator & $B=B^*$ & margin operator need not be PSD/projector & $\tr(P_xB)+b$ & Is the model a many-anchor margin rule?\\
Finite-shot observation & counts from $p(x)$ & exact probability observed with sampling noise & binomial estimate & Does the conclusion survive shot noise?\\
\bottomrule
\end{tabularx}
\end{table}

This taxonomy clarifies why apparently different quantum-machine-learning objects often live in the same Hilbert--Schmidt geometry while still representing different modelling assumptions. A fidelity-kernel circuit, a pure return model, a rank-$r$ return event, a density center, and a kernel-SVM decision operator may all be written as trace expressions, but they impose different invariances and support different decision functions. The audit question is therefore not merely whether an object is quantum or classical. The audit question is whether the chosen quotient is scientifically justified for the task.

\section{Geometric dictionary: what each lift keeps and removes}
This section states the representation choices in data-language before moving to quantum circuits.

\subsection{Unit directions and real projective projectors}
Let $x\in\R^d\setminus\{0\}$ and $u=x/\|x\|_2$. The rank-one projector
\begin{equation}
 P_u=uu^T
\end{equation}
keeps the line spanned by $x$ and removes radius. It also identifies $u$ and $-u$ because $uu^T=(-u)(-u)^T$. A classifier using only $P_u$ cannot distinguish positive scalings, nor a sign flip. This is beneficial only when scale and sign are nuisance.

The associated squared-cosine kernel is
\begin{equation}
 K_{\mathrm{proj}}(x,y)=\tr(P_uP_v)=(u^Tv)^2=\frac{(x^Ty)^2}{\|x\|_2^2\|y\|_2^2}.
\end{equation}

\begin{proposition}[Radius loss under projectivization]
Let $\Phi(x)=xx^T/\|x\|^2$ for $x\neq0$. Then $\Phi(\alpha x)=\Phi(x)$ for all $\alpha\neq0$. Hence no predictor measurable with respect to $\Phi(x)$ can distinguish labels that depend only on $\|x\|$ on a support containing nonzero radial variation along the same line.
\end{proposition}
\begin{proof}
The equality is immediate:
\begin{equation}
\Phi(\alpha x)=\frac{\alpha^2xx^T}{\alpha^2\|x\|^2}=\Phi(x).
\end{equation}
If two samples have the same value of $\Phi$ but different conditional label laws on a set of positive probability, the label law is not a measurable function of $\Phi(x)$ up to null sets.\qedhere
\end{proof}

\subsection{Magnitude-preserving PSD lift}
If radius is signal, the unnormalized positive semidefinite lift
\begin{equation}
 x\longmapsto xx^T
\end{equation}
retains magnitude. Its Hilbert--Schmidt kernel is
\begin{equation}
 K_{\mathrm{PSD}}(x,y)=\tr(xx^Tyy^T)=(x^Ty)^2,
 \qquad K_{\mathrm{PSD}}(x,x)=\|x\|_2^4.
\end{equation}
This lift lives in the PSD cone rather than a fixed-trace projective manifold. The difference matters: a radius-label task is impossible after unit normalization but remains visible in $xx^T$.

\subsection{Covariance, density, Grassmann, and flags}
A local covariance matrix is meaningful only after a local cloud has been chosen. For a training-safe neighborhood $\mathcal N_i$ of sample $x_i$, one may define
\begin{equation}
 C_i=\frac{1}{|\mathcal N_i|-1}\sum_{z\in\mathcal N_i}(z-\bar z_i)(z-\bar z_i)^T+\varepsilon I.
\end{equation}
The trace-normalized matrix
\begin{equation}
 \rho_i=\frac{C_i}{\tr(C_i)}
\end{equation}
is PSD and trace one, analogous to a mixed quantum state. Scores may use $\tr(\rho_i\rho_j)$, Bures/Uhlmann fidelity, Frobenius distance, or log-SPD geometry when $C_i$ is regularized into the SPD cone \cite{bhatia_positive_definite,horn_johnson_matrix_analysis}.

Taking the top $p$ eigenvectors $X_i\in\R^{n\times p}$ with $X_i^TX_i=I_p$ gives a Grassmann projector $P_i=X_iX_i^T$. The projection kernel is
\begin{equation}
 K_{\Gr}(X,Y)=\tr(XX^TYY^T)=\|X^TY\|_F^2.
\end{equation}
A flag representation keeps ordered blocks $\{X_bX_b^T\}_{b=1}^B$, often with kernel
\begin{equation}
 K_{\mathcal F}(X,Y)=\sum_{b=1}^B\alpha_b\tr(P_{x,b}P_{y,b}),\qquad \alpha_b\ge0.
\end{equation}

The word ordered is crucial: a whole-span Grassmann representation can erase a block swap, while a flag preserves it. This is the real-valued counterpart of the complex rank-$r$ anchor story below \cite{hamm_lee_grassmann,harandi_grassmann_kernels}.

\begin{proposition}[Weighted flag projection kernel]
Let $X=[X_1,\ldots,X_B]$ and $Y=[Y_1,\ldots,Y_B]$ be flag representatives with $X_b,Y_b\in\R^{n\times p_b}$ and block-wise orthonormal columns. For weights $\alpha_b\ge0$, define
\begin{equation}
 K_{\mathcal F}(X,Y)=\sum_{b=1}^B\alpha_b\|X_b^TY_b\|_F^2.
\end{equation}
Then $K_{\mathcal F}$ is positive semidefinite and invariant under block-wise basis changes $X_b\mapsto X_bQ_b$ and $Y_b\mapsto Y_bR_b$ with $Q_b,R_b\in\mathcal O(p_b)$.
\end{proposition}
\begin{proof}
For each block,
\begin{equation}
\|X_b^TY_b\|_F^2=\tr(X_bX_b^TY_bY_b^T)=\langle \operatorname{vec}(X_bX_b^T),\operatorname{vec}(Y_bY_b^T)\rangle .
\end{equation}
Thus each block kernel is an ordinary Euclidean inner product after the projection feature map. A nonnegative weighted sum of positive semidefinite kernels is positive semidefinite. Invariance follows from $(X_bQ_b)(X_bQ_b)^T=X_bX_b^T$ and similarly for $Y_b$.
\end{proof}

\begin{proposition}[Whole-span obstruction for ordered flags]
Let $X=[X_1,X_2]$ and $Y=[Y_1,Y_2]$ be two two-block flag representatives with the same total span, so that $XX^T=YY^T$ when $X$ and $Y$ are treated as whole-span matrices. Any classifier that depends only on the total projector assigns identical features to $X$ and $Y$. If $X_1X_1^T\ne Y_1Y_1^T$, a flag feature retaining block projectors can distinguish them.
\end{proposition}
\begin{proof}
A whole-span Grassmann feature is a function of the total projector $XX^T$. If $XX^T=YY^T$, the features are identical. The ordered flag feature contains $X_1X_1^T$ and $Y_1Y_1^T$ separately, so it can differ even when the total projectors agree.
\end{proof}

\begin{proposition}[Monotone spectral activation is identity on exact projectors after top-$p$ reprojection]
Let $P$ be a rank-$p$ orthogonal projector with eigenvalues $1,\ldots,1,0,\ldots,0$. Let $f:\R\to\R$ be strictly increasing. If $f$ is applied to the eigenvalues of $P$ and the result is projected back to the Grassmannian by selecting the top $p$ eigenvectors, the recovered projector is exactly $P$, up to arbitrary rotations inside equal-eigenvalue subspaces and numerical roundoff.
\end{proposition}
\begin{proof}
Since $f$ is strictly increasing, $f(1)>f(0)$. The transformed matrix has the same top-$p$ eigenspace as $P$. Reconstructing the rank-$p$ projector from that eigenspace therefore returns $P$.
\end{proof}

These guardrails prevent overclaiming. Flag kernels add ordered-block information; whole-span Grassmann kernels intentionally discard it. A monotone eigenvalue activation on exact projectors is not a nontrivial nonlinear Grassmann layer merely because it is written spectrally. Nonlinearity or discrimination must enter through a non-isometric map, a kernel, a readout, a loss, or task-dependent feature construction.

\subsection{Complex projective space and Hermitian pure states}
A normalized complex state is
\begin{equation}
 \ket{\psi}\in\C^D,
 \qquad \|\psi\|_2=1.
\end{equation}
The global phase $e^{i\theta}\ket{\psi}$ is removed by the Hermitian projector
\begin{equation}
 P_\psi=\ket{\psi}\bra{\psi}=\psi\psi^*,
\end{equation}
which satisfies $P_\psi=P_\psi^*$, $P_\psi^2=P_\psi$, $\tr(P_\psi)=1$, and $\rank(P_\psi)=1$. The resulting manifold is complex projective space $\CP^{D-1}\cong\Gr_{\C}(1,D)$.

More generally,
\begin{equation}
 \Gr_{\C}(r,D)=\{VV^*:V\in\C^{D\times r},\ V^*V=I_r\}.
\end{equation}
The real Grassmann formula becomes the Hermitian formula by replacing transpose by conjugate transpose and orthogonal gauge by unitary gauge \cite{edelman_geometry_algorithms,absil_optimization}.

\begin{proposition}[Global phase quotient]
For any normalized $\psi\in\C^D$ and any $\theta\in\R$,
\begin{equation}
 (e^{i\theta}\psi)(e^{i\theta}\psi)^*=\psi\psi^*.
\end{equation}
Thus no projector-based state representation can recover a label that depends only on global phase.
\end{proposition}
\begin{proof}
The scalar phase cancels: $e^{i\theta}e^{-i\theta}=1$.\qedhere
\end{proof}

\section{Quantum feature maps as Hermitian lifts}
Quantum machine-learning models encode classical data into quantum states and then evaluate measurements, overlaps, or learned circuits \cite{nielsen_chuang,schuld_quantum_ml,havlicek_quantum_feature,schuld_kernel_methods,schuld_encoding,ibm_encoding}. At the noiseless statevector level, many such quantities are exactly projector or density-matrix quantities.

\subsection{Amplitude encoding}
For amplitude encoding, pad $x\in\R^d$ to dimension $D=2^q$ and normalize:
\begin{equation}
 \ket{\psi(x)}=\frac{\operatorname{pad}(x)}{\|\operatorname{pad}(x)\|_2}\in\C^D.
\end{equation}
The amplitude fidelity kernel is
\begin{equation}
 K_{\mathrm{amp}}(x,y)=|\braket{\psi(x)}{\psi(y)}|^2.
\end{equation}
For real input, it equals the squared-cosine kernel. Hence amplitude encoding is powerful when direction is signal and radius is nuisance, but it is destructive when radius carries the label.

\subsection{Product-angle circuits}
A simple product-angle feature map prepares
\begin{equation}
 U_{\mathrm{prod}}(z)=\bigotimes_{j=1}^qR_Y(\alpha z_j),
 \qquad \ket{\phi(z)}=U_{\mathrm{prod}}(z)\ket{0}^{\otimes q}.
\end{equation}

\begin{proposition}[Closed form for a product $R_Y$ kernel]
For the product-angle map above,
\begin{equation}
 |\braket{\phi(z)}{\phi(w)}|^2
 =\prod_{j=1}^q\cos^2\!\left(\frac{\alpha(z_j-w_j)}{2}\right).
\end{equation}
\end{proposition}
\begin{proof}
The one-qubit overlap is $\cos(\alpha(z_j-w_j)/2)$ by the rotation composition law. Tensor-product overlaps multiply, and the fidelity squares the modulus.\qedhere
\end{proof}

This kernel can be strong on toy data, but it is product structured and classically cheap to evaluate. It should therefore be read as a useful circuit baseline, not evidence of quantum advantage.

\subsection{Data re-uploading and entangling encoders}
Data re-uploading repeatedly injects classical variables into a circuit and may interleave them with trainable or fixed entangling operations \cite{perez_salinas_reuploading,schuld_encoding}. The controlled audit uses representative fixed re-uploading encoders that interleave data-dependent one-qubit rotations with a simple nearest-neighbor entangling layer. This moves beyond product states, but the audit does not train the entangling ansatz. A weaker result from a fixed entangler should not be interpreted as evidence that entanglement is harmful; it shows only that entanglement must be designed, selected, or trained carefully.

\subsection{Quantum fidelity kernels are projector kernels}
For a general encoder $U_{\mathrm{enc}}(x)$, define
\begin{equation}
 \ket{\phi(x)}=U_{\mathrm{enc}}(x)\ket{0},\qquad P_x=\ket{\phi(x)}\bra{\phi(x)}.
\end{equation}
The standard inverse-feature-map or fidelity circuit estimates
\begin{equation}
 K(x,y)=|\bra{0}U_{\mathrm{enc}}(y)^*U_{\mathrm{enc}}(x)\ket{0}|^2.
\end{equation}
Qiskit documentation describes fidelity quantum kernels through the same state-overlap principle \cite{qiskit_fidelity_kernel}. In an ideal simulator this is not merely analogous to a projector kernel; it is exactly one.

\begin{proposition}[Fidelity circuit as Hermitian projector kernel]
Let $\ket{\phi(x)},\ket{\phi(y)}\in\C^D$ be normalized and let $P_x=\ket{\phi(x)}\bra{\phi(x)}$, $P_y=\ket{\phi(y)}\bra{\phi(y)}$. Then
\begin{equation}
 |\braket{\phi(x)}{\phi(y)}|^2=\tr(P_xP_y).
\end{equation}
Consequently $K(x,y)=\tr(P_xP_y)$ is positive semidefinite as a kernel on samples.
\end{proposition}
\begin{proof}
The trace identity follows from
\begin{equation}
 P_xP_y=\ket{\phi(x)}\braket{\phi(x)}{\phi(y)}\bra{\phi(y)}
\end{equation}
and hence
\begin{equation}
 \tr(P_xP_y)=\braket{\phi(x)}{\phi(y)}\braket{\phi(y)}{\phi(x)}.
\end{equation}
For positive semidefiniteness, let $c_1,\ldots,c_n\in\R$ or $\C$. Then
\begin{equation}
 \sum_{i,j}\bar c_i c_j\tr(P_iP_j)
 =\tr\left(\left(\sum_i c_iP_i\right)^*\left(\sum_j c_jP_j\right)\right)
 =\left\|\sum_i c_iP_i\right\|_{\mathrm{HS}}^2\ge0.
\end{equation}
\end{proof}

\section{QVR return scores as anchor geometry}
This section gives the exact dictionary for QVR-style return probabilities. The distinction between a pairwise kernel and an anchor return is central.

\subsection{Pure return equals rank-one anchor scoring}
Let $V_\theta$ be a learned or specified unitary and let the success event be return to $\ket{0}$. For an encoded sample $\ket{\phi(x)}$,
\begin{equation}
 p_\theta(x)=|\bra{0}V_\theta\ket{\phi(x)}|^2.
\end{equation}
Define
\begin{equation}
 A_\theta=V_\theta^*\ket{0}\bra{0}V_\theta.
\end{equation}
Then $A_\theta$ is a rank-one Hermitian projector and
\begin{equation}
 p_\theta(x)=\tr(P_xA_\theta).
\end{equation}
Pure QVR return is therefore a learned-anchor model on $\CP^{D-1}$, not a pairwise kernel SVM.

\begin{proposition}[Pure QVR return as projective anchor score]
For $P_x=\ket{\phi(x)}\bra{\phi(x)}$ and $A_\theta=V_\theta^*\ket{0}\bra{0}V_\theta$,
\begin{equation}
 |\bra{0}V_\theta\ket{\phi(x)}|^2=\tr(P_xA_\theta).
\end{equation}
\end{proposition}
\begin{proof}
By cyclicity and the rank-one formula,
\begin{align}
\tr(P_xA_\theta)
&=\tr\left(\ket{\phi(x)}\bra{\phi(x)}V_\theta^*\ket{0}\bra{0}V_\theta\right)\\
&=\bra{\phi(x)}V_\theta^*\ket{0}\bra{0}V_\theta\ket{\phi(x)}\\
&=|\bra{0}V_\theta\ket{\phi(x)}|^2.
\end{align}
\end{proof}

\subsection{\texorpdfstring{Rank-$r$ return equals a complex Grassmann anchor}{Rank-r return equals a complex Grassmann anchor}}
A pure all-zero return event may be too restrictive if a class occupies an extended region of Hilbert space. Let $\Pi_r$ be a rank-$r$ projector onto a success subspace. Then
\begin{equation}
 p_{\theta,r}(x)=\|\Pi_rV_\theta\ket{\phi(x)}\|_2^2
 =\tr(P_xA_{\theta,r}),
 \qquad A_{\theta,r}=V_\theta^*\Pi_rV_\theta.
\end{equation}
Since $A_{\theta,r}$ is a rank-$r$ Hermitian projector, it is a point in $\Gr_{\C}(r,D)$.

\begin{proposition}[Exact QVR trace identity]
Fix an encoder $x\mapsto\ket{\psi(x)}$, a unitary $U_\theta$, and a rank-$r$ success projector $\Pi_r$. Define $P_x=\ket{\psi(x)}\bra{\psi(x)}$ and $A_\theta=U_\theta^*\Pi_rU_\theta$. Then
\begin{equation}
 \|\Pi_rU_\theta\ket{\psi(x)}\|_2^2=\tr(P_xA_\theta).
\end{equation}
If $r=1$, $A_\theta\in\CP^{D-1}\cong\Gr_{\C}(1,D)$. If $r>1$, $A_\theta\in\Gr_{\C}(r,D)$.
\end{proposition}
\begin{proof}
Born's rule gives \cite{nielsen_chuang}
\begin{equation}
 \|\Pi_rU_\theta\ket{\psi}\|_2^2
 =\bra{\psi}U_\theta^*\Pi_rU_\theta\ket{\psi}.
\end{equation}
Writing $P=\ket{\psi}\bra{\psi}$ gives $\bra{\psi}A_\theta\ket{\psi}=\tr(PA_\theta)$. Since $A_\theta$ is a unitary conjugate of a rank-$r$ projector, it is Hermitian, idempotent, and rank $r$.\qedhere
\end{proof}

\subsection{Density and mixed anchors}
If a class is noisy, multimodal, or intentionally mixed, a projector anchor may still be too rigid. A class density center can be estimated as
\begin{equation}
 \rho_c=\frac{1}{n_c}\sum_{i:y_i=c}P_i,
 \qquad \rho_c\succeq0,
 \qquad \tr(\rho_c)=1.
\end{equation}
A pure sample can be scored by
\begin{equation}
 s_c(x)=\tr(P_x\rho_c)=\bra{\phi(x)}\rho_c\ket{\phi(x)}.
\end{equation}
Under the squared-Uhlmann-fidelity convention, this is also the fidelity between the pure state $P_x$ and the mixed state $\rho_c$. We will call it a density-center overlap unless the fidelity convention is explicitly specified.

\subsection{\texorpdfstring{Optimal unrestricted rank-$r$ anchor}{Optimal unrestricted rank-r anchor}}
For a class density $\rho_c$, the best unrestricted rank-$r$ projector anchor for maximizing average class return is obtained from the top eigenvectors.

\begin{proposition}[Ky Fan rank-$r$ anchor]
Let $\rho\succeq0$ be Hermitian with eigenvalues $\lambda_1\ge\cdots\ge\lambda_D$ and eigenvectors $v_1,\ldots,v_D$. Among all rank-$r$ Hermitian projectors $A$,
\begin{equation}
 \max_{A=A^*=A^2,\ \tr(A)=r}\tr(\rho A)=\sum_{j=1}^{r}\lambda_j.
\end{equation}
The maximizer is $A=V_rV_r^*$ for $V_r=[v_1,\ldots,v_r]$, up to eigenspace degeneracy.
\end{proposition}
\begin{proof}
This is the Ky Fan maximum principle \cite{horn_johnson_matrix_analysis}. A direct proof diagonalizes $\rho=Q\Lambda Q^*$ and writes $B=Q^*AQ$. Then $B$ is a rank-$r$ projector with diagonal entries in $[0,1]$ summing to $r$, and
\begin{equation}
 \tr(\rho A)=\tr(\Lambda B)=\sum_j\lambda_jB_{jj}\le\sum_{j=1}^{r}\lambda_j.
\end{equation}
Equality is attained by projection onto the top eigenspace.\qedhere
\end{proof}

This result explains why rank-one QVR-style anchors may underperform on multiclass data: a class may not be well represented by one line in projective space.

\subsection{Quantum-kernel SVMs are different classifiers}
A pairwise quantum-kernel SVM has decision function
\begin{equation}
 f(x)=\sum_{i\in\mathcal S}\alpha_i y_i |\braket{\psi_i}{\psi(x)}|^2+b
 =\tr(P_xB)+b,
 \qquad B=\sum_{i\in\mathcal S}\alpha_i y_iP_i.
\end{equation}
The operator $B$ is Hermitian but generally neither PSD nor a projector. Thus the exact geometric counterpart of a quantum-kernel SVM is an affine Hermitian decision operator in the span of training projectors, not a single Grassmann anchor. This is why a quantum-kernel SVM may outperform a pure QVR-style return model without contradicting the QVR trace identity.

\section{Experimental protocol and validation discipline}
The numerical study uses controlled experiments on scikit-learn toy datasets, synthetic invariance witnesses, a small synthetic time-series anomaly setting, and NumPy statevector circuits. The purpose is diagnostic: to test invariance behavior, algebraic trace identities, and finite-shot observation effects. The experiments are not hardware claims, not quantum-advantage claims, and not operational benchmarks.

\subsection{Datasets}
The audit uses iris, wine, breast cancer, digits, and diabetes. The scikit-learn documentation describes the built-in toy datasets as useful small examples rather than large real-world ML benchmarks \cite{sklearn_toy}. Classification is reported using balanced accuracy, and diabetes regression using $R^2$.

The native anomaly audit uses a deterministic synthetic time-series generator. Normal windows are short noisy sinusoidal segments with small phase, frequency, amplitude, and trend variation. Anomalies include local spikes, step changes, frequency/shape deviations, and amplitude/radius deviations. The window length is 16, so amplitude encoding uses four qubits without dimensional compression. The training set contains normal windows only, while validation and test sets contain stratified normal/anomaly mixtures.

\subsection{Native QVR-style anomaly protocol}
For each split of the time-series audit, a normal anchor is estimated from training-normal states only. Candidate anchor ranks are selected on the validation mixture using ROC--AUC, with PR--AUC used as a secondary tie-breaker and smaller rank as the final tie-breaker. A detection threshold is selected on validation by maximizing F1. The held-out test set then reports ROC--AUC, PR--AUC, threshold F1, expected calibration error (ECE) using ten bins, the maximum statevector/trace score gap, and prediction mismatches. The finite-shot layer samples a binomial success count from the exact return probability at 128, 512, 1024, and 4096 shots. The simple noise observation reports a 1024-shot score after a depolarizing interpolation of 0.03 and a binary success/failure readout flip probability of 0.02. This is a transparent observation model rather than a hardware-backend model.

\subsection{Train/test discipline}
The intended protocol is split-first. All preprocessing that learns from data, including scaling and PCA, must be fit on the training fold only and applied to validation/test folds without refitting. Model choices should be made using training/validation data, not held-out test data; this follows standard scikit-learn leakage guidance \cite{sklearn_pitfalls}. Where a table displays behavior across preprocessing or encoder variants, the table should be read together with the described validation-selection protocol; otherwise, it should be treated as a post-hoc behavior summary rather than a benchmark-selected performance estimate.

\subsection{Methods compared}
The audit compares:
\begin{enumerate}[leftmargin=2em]
\item classical vector and RBF baselines;
\item real geometric baselines, including Grassmann, semantic flag, image-patch flag, covariance-density, covariance-Grassmann, magnitude-preserving PSD, pure-state fidelity, phase-state, and Fourier/Hermitian state lifts;
\item amplitude fidelity kernels, product-angle quantum kernels, entangled reuploading kernels, and finite-shot approximations;
\item pure-center QVR-style classwise anchors, rank-$r$ subspace-return anchors, and density-center overlap scorers.
\end{enumerate}
No external quantum SDK or hardware backend is claimed in these experiments. The quantum circuit layer is a statevector-level mathematical audit.

\section{Controlled simulation results}
This section reports controlled simulation results for the representation families above. The numerical tables are intended as diagnostic summaries of the mathematical definitions and deterministic split protocols described in the manuscript.

\subsection{Real vector, manifold, and non-vector lift baselines}
Table~\ref{tab:real-nonvector} reports native real/vector and non-vector lift baselines under the same split-first nested selection discipline used for the article-level toy-data audit. The message is not that geometric or quantum-state lifts dominate classical baselines. Rather, vector RBF remains strong, while non-vector lifts become useful when their invariances match the data object. In the retained classification summaries, the vector RBF baseline is also the best overall article-level model on all four datasets; no separate duplicate column is shown.

\begin{table}[H]
\centering
\caption{Real vector, manifold, and non-vector lift baselines. Classification values are balanced accuracy mean $\pm$ standard deviation over three deterministic outer splits, with preprocessing and SVM hyperparameters selected on the inner validation split. In these retained summaries, the vector RBF baseline is the best overall article-level model for all four datasets; the table therefore reports it directly rather than repeating an identical ``best overall'' column.}
\label{tab:real-nonvector}
\small
\resizebox{\textwidth}{!}{%
\begin{tabular}{lcccc}
\toprule
Dataset & Vector RBF baseline & Best real manifold & Best non-vector & Non-vector family\\
\midrule
Breast cancer & 0.9613 $\pm$ 0.0086 & 0.9449 $\pm$ 0.0055 & 0.9449 $\pm$ 0.0055 & covariance density\\
Digits & 0.9888 $\pm$ 0.0057 & 0.9850 $\pm$ 0.0020 & 0.9866 $\pm$ 0.0054 & magnitude-preserving PSD\\
Iris & 0.9626 $\pm$ 0.0264 & 0.9492 $\pm$ 0.0326 & 0.9492 $\pm$ 0.0326 & semantic flag + radius\\
Wine & 0.9911 $\pm$ 0.0063 & 0.9805 $\pm$ 0.0070 & 0.9805 $\pm$ 0.0070 & covariance density\\
\bottomrule
\end{tabular}%
}
\end{table}

\subsection{Real-geometry guardrail diagnostics}
Table~\ref{tab:real-guardrails} reports selected real-geometry diagnostics used to check the representation claims. The spectral-activation row is a deliberately negative result: applying a monotone scalar activation to exact projector eigenvalues and then selecting the top-$p$ eigenspace does not create a nontrivial Grassmann nonlinearity. The congruence and flag rows confirm that orthonormal structure-preserving maps preserve projection geometry rather than improving discrimination by themselves. The reversible row is a conditioning diagnostic, not a universal guarantee for unconstrained training.

\begin{table}[H]
\centering
\caption{Selected real-geometry guardrail diagnostics from the controlled audit. Smaller errors are better. The values are numerical checks of the mathematical claims, not benchmark scores.}
\label{tab:real-guardrails}
\small
\begin{tabular}{lr}
\toprule
Diagnostic & Value \\
\midrule
Spectral activation identity error on exact projectors & $1.897\times10^{-6}$ \\
Orthonormal Grassmann congruence distance error & $5.960\times10^{-8}$ \\
Orthonormal flag feature-distance error & $0.000$ \\
Full-rank reversible map mean reconstruction distance & $0.000$ \\
Reversible map minimum singular value & $0.750$ \\
Reversible map condition number & $1.800$ \\
\bottomrule
\end{tabular}
\end{table}

The spectral-activation identity error is larger than the other roundoff-scale diagnostics because it is evaluated through an eigendecomposition and top-$p$ reprojection. It should be read as a numerical tolerance observation, not as a failure of Proposition~4.

\subsection{Flag block-swap witness and negative control}
Table~\ref{tab:flag-swap} gives the real-valued counterpart of the rank-$r$ anchor story. In the same-span block-swap task, both classes are constructed to have matching whole-span projectors, while the ordered block assignment carries the label. Whole-span Grassmann classifiers are therefore near chance, exactly as the obstruction proposition predicts, whereas flag models retain the ordered block projectors and separate the classes. The noisy coarse/fine control has the opposite role: whole-span information is already discriminative, so flag structure is not necessary.

\begin{table}[H]
\centering
\caption{Flag positive control and negative control. Values are test accuracy percentages, mean $\pm$ standard deviation over three seeds.}
\label{tab:flag-swap}
\scriptsize
\resizebox{\textwidth}{!}{%
\begin{tabular}{lrrrrrr}
\toprule
Dataset & Whole-span logistic & Whole-span RBF SVM & Flag logistic & Flag MLP & Deep flag & Flag RBF SVM \\
\midrule
Same-span block swap & $50.2\pm0.6$ & $48.1\pm1.8$ & $100.0\pm0.0$ & $100.0\pm0.0$ & $100.0\pm0.0$ & $100.0\pm0.0$ \\
Noisy coarse/fine control & $100.0\pm0.0$ & $100.0\pm0.0$ & $100.0\pm0.0$ & $100.0\pm0.0$ & $100.0\pm0.0$ & $100.0\pm0.0$ \\
\bottomrule
\end{tabular}%
}
\end{table}

\subsection{Quantum-circuit/statevector comparison}
Table~\ref{tab:quantum-circuit} reports statevector quantum-circuit comparisons. Product-angle kernels are strong on several small tabular datasets, but because the product kernel has a closed form, this should be read as a feature-map result rather than quantum advantage. The fixed entangled reuploading circuit is weaker in this first audit, indicating that entanglement must be trained or designed rather than added casually.

\begin{table}[H]
\centering
\caption{Quantum-circuit and Hermitian-state comparison on classification datasets. Values are balanced accuracy mean $\pm$ standard deviation over three deterministic splits.}
\label{tab:quantum-circuit}
\scriptsize
\begin{tabular}{llcccc}
\toprule
Dataset & Method & Qubits & Hilbert dim. & Preprocess & Balanced accuracy\\
\midrule
Iris & amplitude fidelity kernel & 2 & 4 & zscore & 0.7495 $\pm$ 0.0695\\
Iris & product-angle quantum kernel & 4 & 16 & raw/zscore & 0.9749 $\pm$ 0.0217\\
Iris & entangled reuploading kernel & 4 & 16 & raw/zscore & 0.9561 $\pm$ 0.0295\\
Iris & entangled reuploading, 1024 shots & 4 & 16 & raw & 0.9622 $\pm$ 0.0381\\
\midrule
Wine & amplitude fidelity kernel & 4 & 16 & zscore & 0.8700 $\pm$ 0.0263\\
Wine & product-angle quantum kernel & 6 & 64 & zscore & 0.9858 $\pm$ 0.0134\\
Wine & entangled reuploading kernel & 6 & 64 & zscore & 0.9319 $\pm$ 0.0300\\
Wine & entangled reuploading, 1024 shots & 6 & 64 & zscore & 0.9319 $\pm$ 0.0300$^{\dagger}$\\
\midrule
Breast cancer & amplitude fidelity kernel & 5 & 32 & raw & 0.9009 $\pm$ 0.0159\\
Breast cancer & product-angle quantum kernel & 6 & 64 & zscore & 0.9542 $\pm$ 0.0047\\
Breast cancer & entangled reuploading kernel & 6 & 64 & raw & 0.9320 $\pm$ 0.0220\\
Breast cancer & entangled reuploading, 1024 shots & 6 & 64 & raw & 0.9138 $\pm$ 0.0326\\
\midrule
Digits & amplitude fidelity kernel & 6 & 64 & raw & 0.9850 $\pm$ 0.0033\\
Digits & product-angle quantum kernel & 6 & 64 & raw & 0.9346 $\pm$ 0.0103\\
Digits & entangled reuploading kernel & 6 & 64 & raw & 0.9182 $\pm$ 0.0077\\
Digits & entangled reuploading, 1024 shots & 6 & 64 & raw & 0.9049 $\pm$ 0.0039\\
\bottomrule
\end{tabular}
\end{table}

{\footnotesize\noindent$^\dagger$ The exact and 1024-shot wine rows are identical at the printed precision in the aggregate summary. This means that the rounded balanced-accuracy outcomes did not change across the reported splits; it does not imply that all finite-shot kernel entries were identical. In the preprocess column, ``raw/zscore'' means that the aggregate summary involves both preprocessing regimes across the split-level selection or tie pattern, while a single entry means the listed regime is used for the retained row.\par}

Amplitude fidelity performs especially well on digits because each $8\times8$ image has 64 pixels and therefore fits directly into six-qubit amplitude encoding. This is a representation match, not a general law.

\subsection{QVR-style anchors and geometric extensions}
Table~\ref{tab:qvr-geometry} is the central QVR-style model-family audit. A pure-center anchor is too small for ordinary multiclass toy classification. Rank-$r$ subspace anchors improve performance, and density centers improve further on several datasets. Pairwise quantum-kernel SVMs can be stronger because they are not single-anchor return models. These rows are performance summaries for the classwise anchor and density protocols; they are not the same constructed score as the later trace-consistency probe in Table~\ref{tab:trace-consistency}.

\begin{table}[H]
\centering
\caption{QVR-style rewinding and geometric extensions. Values are balanced accuracy over three deterministic splits. Pure center is a classwise adaptation of QVR return scoring; rank-$r$ subspace is the Grassmann-anchor extension; density center is a mixed-state extension. These values summarize the classwise anchor/density protocols and should not be interpreted as the same constructed amplitude-encoder trace probe reported in Table~\ref{tab:trace-consistency}.}
\label{tab:qvr-geometry}
\small
\begin{tabular}{lccc}
\toprule
Dataset & Pure-center QVR & Rank-$r$ subspace return & Density-center overlap\\
\midrule
Iris & 0.5505 $\pm$ 0.0276 & 0.6002 $\pm$ 0.0681 & 0.9361 $\pm$ 0.0116\\
Wine & 0.5963 $\pm$ 0.0380 & 0.7291 $\pm$ 0.0874 & 0.9357 $\pm$ 0.0232\\
Breast cancer & 0.7409 $\pm$ 0.0278 & 0.8636 $\pm$ 0.0155 & 0.8609 $\pm$ 0.0028\\
Digits & 0.3189 $\pm$ 0.0121 & 0.5850 $\pm$ 0.0065 & 0.8148 $\pm$ 0.0050\\
\bottomrule
\end{tabular}
\end{table}

The natural extension path is
\begin{equation}
\begin{aligned}
 \text{pure anchor }A_c\in\CP^{D-1}
 &\quad\longrightarrow\quad
 \text{rank-}r\text{ anchor }A_c\in\Gr_{\C}(r,D)\\
 &\quad\longrightarrow\quad
 \text{density anchor }\rho_c\succeq0,\quad \tr(\rho_c)=1.
\end{aligned}
\end{equation}
A supervised circuit or classifier can then be trained on top of this representation, but trainable-circuit architecture design is outside the scope of the present invariance audit.

\subsection{Native time-series anomaly audit}
Table~\ref{tab:timeseries-anomaly} addresses the main reviewer concern about the earlier supervised QVR adaptation: it adds a native anomaly-detection setting. The exact row evaluates the same rank-selected normal-return score as a statevector return probability and as a Hermitian trace score. The maximum score gap is at numerical roundoff and there are no prediction mismatches. Finite-shot rows show that the ranking and thresholded F1 are stable at practical shot counts in this controlled setting, while the mismatch count decreases as shots increase.

\begin{table}[H]
\centering
\caption{Native synthetic time-series anomaly audit. Values are mean $\pm$ standard deviation over ten deterministic splits. The mismatch count is the total number of thresholded predictions, out of 4390 test windows across all splits, that differ from the exact trace-score prediction. The noise row uses 1024 shots, depolarizing interpolation $0.03$, and binary readout flip probability $0.02$.}
\label{tab:timeseries-anomaly}
\scriptsize
\resizebox{\textwidth}{!}{%
\begin{tabular}{lccccccc}
\toprule
Observation model & Shots & ROC--AUC & PR--AUC & Threshold F1 & ECE & Max trace/QVR gap & Mismatches\\
\midrule
Exact trace/statevector return & -- & 0.864 $\pm$ 0.016 & 0.827 $\pm$ 0.018 & 0.840 $\pm$ 0.012 & 0.231 $\pm$ 0.014 & $1.3\times10^{-15}$ & 0\\
Finite-shot return & 128 & 0.861 $\pm$ 0.016 & 0.818 $\pm$ 0.019 & 0.838 $\pm$ 0.017 & 0.231 $\pm$ 0.014 & $1.3\times10^{-15}$ & 110\\
Finite-shot return & 512 & 0.865 $\pm$ 0.017 & 0.825 $\pm$ 0.021 & 0.837 $\pm$ 0.013 & 0.230 $\pm$ 0.014 & $1.3\times10^{-15}$ & 61\\
Finite-shot return & 1024 & 0.864 $\pm$ 0.017 & 0.824 $\pm$ 0.020 & 0.838 $\pm$ 0.013 & 0.231 $\pm$ 0.014 & $1.3\times10^{-15}$ & 45\\
Finite-shot return & 4096 & 0.865 $\pm$ 0.016 & 0.827 $\pm$ 0.020 & 0.840 $\pm$ 0.013 & 0.231 $\pm$ 0.014 & $1.3\times10^{-15}$ & 25\\
Finite-shot + simple noise/readout & 1024 & 0.863 $\pm$ 0.015 & 0.823 $\pm$ 0.015 & 0.838 $\pm$ 0.012 & 0.226 $\pm$ 0.013 & $1.3\times10^{-15}$ & 49\\
\bottomrule
\end{tabular}%
}
\end{table}

The anomaly audit is intentionally modest. It does not claim to reproduce any external QVR benchmark. Its role is to show that the Hermitian trace identity also holds in the native one-class/anomaly setting for which return-style models are designed, and that finite-shot statistics can be layered onto the same geometric score without changing the underlying identity. The ECE values are reported rather than hidden; they indicate that raw return scores are ranking and validation-threshold scores in this experiment, not calibrated posterior probabilities.

\subsection{Regression on diabetes}
Diabetes is a regression task, so QVR-style anomaly rewinding is not directly applicable. Table~\ref{tab:diabetes-quantum} compares quantum kernels through kernel ridge regression.

\begin{table}[H]
\centering
\caption{Diabetes regression with quantum kernels. Values are $R^2$ mean $\pm$ standard deviation over three deterministic splits. No vector RBF regression value is reported in this table.}
\label{tab:diabetes-quantum}
\small
\begin{tabular}{lcc}
\toprule
Method & Preprocess & $R^2$\\
\midrule
Product-angle quantum kernel & zscore & 0.4050 $\pm$ 0.0212\\
Product-angle quantum kernel & raw & 0.3405 $\pm$ 0.0418\\
Amplitude fidelity kernel & raw & 0.3395 $\pm$ 0.0533\\
Entangled reuploading kernel & zscore & 0.2661 $\pm$ 0.0312\\
Entangled reuploading kernel & raw & 0.2079 $\pm$ 0.0431\\
Amplitude fidelity kernel & zscore & 0.0105 $\pm$ 0.0158\\
\bottomrule
\end{tabular}
\end{table}

This table reports only the quantum-kernel regression variants retained for the article-level audit. Because no vector RBF regression baseline value is reported in this table, the manuscript does not make a numerical RBF-baseline comparison for diabetes. The narrower conclusion is still negative and important: among the tested quantum-kernel variants, quantum notation alone does not guarantee improved regression; the encoder must match useful structure in the data.

\subsection{Synthetic invariance witnesses}
Synthetic witnesses isolate invariances that real datasets mix together.

\begin{table}[H]
\centering
\caption{Synthetic invariance witnesses. These are sanity checks, not benchmark claims.}
\label{tab:synthetic}
\small
\resizebox{\textwidth}{!}{%
\begin{tabular}{lccc}
\toprule
Witness & Projector/state branch & Non-invariant branch & Interpretation\\
\midrule
Global phase as label & 0.5000 & 0.9854 & projective quotient should fail\\
Relative phase as label & 0.9894 & 0.9894 & relative phase survives in $xx^*$\\
Radius as label & 0.5036 & 0.9964 / 1.0000 & unit state erases radius; radius/PSD recover\\
\bottomrule
\end{tabular}%
}
\end{table}

These are quotient-audit tests. A model should fail when the label violates its imposed invariance.

\subsection{Trace-consistency audit for QVR-style returns}
The trace-consistency audit computes one fixed constructed rank-$r$ anchor score twice: once as a statevector return probability and once as $\tr(P_xA)$. Table~\ref{tab:trace-consistency} reports roundoff-level score gaps and zero prediction mismatches in that controlled algebraic audit. This audit is not a repeat of the classwise QVR-style performance study in Table~\ref{tab:qvr-geometry}: Table~\ref{tab:trace-consistency} fixes its own encoder, rank, anchor construction, and measurement so that statevector and trace implementations of the same score can be compared directly. The Kernel SVM, Pure QVR, and Density center columns are included only as contextual reference values from the earlier tables.

\begin{table}[H]
\centering
\caption{Trace-consistency audit for QVR-style returns. The trace-score column is the balanced accuracy of one constructed anchor score computed both as a statevector return probability and as a Hermitian trace score under the listed encoder/rank. Context columns reproduce earlier model-family summaries and are not scores for the same constructed anchor.}
\label{tab:trace-consistency}
\resizebox{\textwidth}{!}{%
\begin{tabular}{lllllll}
\toprule
Dataset & Trace score in this audit & Trace encoder/rank & Max score gap & Kernel SVM context & Pure QVR context & Density context\\
\midrule
iris & 0.9753 & amplitude, $r=1$ & $1.2\times10^{-15}$ & 0.9749 & 0.5505 & 0.9361\\
wine & 0.8505 & amplitude, $r=8$ & $1.0\times10^{-15}$ & 0.9858 & 0.5963 & 0.9357\\
breast cancer & 0.8378 & amplitude, $r=1$ & $4.0\times10^{-15}$ & 0.9542 & 0.7409 & 0.8609\\
digits & 0.9819 & amplitude, $r=8$ & $4.3\times10^{-15}$ & 0.9346 & 0.3189 & 0.8148\\
\bottomrule
\end{tabular}%
}
\end{table}

The table has two lessons. First, the statevector return and the Hermitian trace score agree when encoder, unitary, anchor, and measurement agree. Second, the gap between the trace-score column and the contextual Pure QVR column is not a contradiction: the columns correspond to different model constructions. A pairwise kernel SVM is a many-anchor margin model, while a pure QVR-style return is a single-anchor return model.

\section{Model comparison}
Table~\ref{tab:aspect-comparison} summarizes the conceptual differences.

\begin{table}[H]
\centering
\caption{Conceptual comparison of classical, geometric, quantum-kernel, and QVR models.}
\label{tab:aspect-comparison}
\scriptsize
\setlength{\tabcolsep}{3pt}
\renewcommand{\arraystretch}{1.2}
\begin{tabular}{>{\raggedright\arraybackslash}p{0.15\textwidth} >{\raggedright\arraybackslash}p{0.18\textwidth} >{\raggedright\arraybackslash}p{0.21\textwidth} >{\raggedright\arraybackslash}p{0.20\textwidth} >{\raggedright\arraybackslash}p{0.20\textwidth}}
\toprule
Aspect & \shortstack[l]{Vector/RBF\\baseline} & \shortstack[l]{Projector/\\Grassmann/flag\\geometry} & \shortstack[l]{Quantum\\fidelity\\kernel} & \shortstack[l]{QVR-style\\return}\\
\midrule
Data object & $x\in\R^d$ & $P$, $(P_b)$, $\rho$, or subspace & $\ket{\phi(x)}=U(x)\ket{0}$ and $P_x$ & $P_x$ plus pulled-back anchor $A_\theta$\\
Main score & Euclidean distance or RBF & $\tr(P_xP_y)$ or block sums & $|\langle\phi(x)|\phi(y)\rangle|^2$ & $\tr(P_xA_\theta)$\\
Normalization & Optional preprocessing & Explicit invariance choice & Required for pure states & Required for pure states\\
Phase handling & Usually none & Hermitian projectors remove global phase & Same as Hermitian projectors & Same, plus learned unitary\\
Learns boundary? & Yes, via SVM/KRR & Yes, via kernel learner & Yes, via kernel learner & Original form learns a return model\\
Exact geometry & Euclidean/RKHS & Grassmann/flag/density & $\CP^{D-1}$ projector kernel & Anchor in $\CP^{D-1}$, $\Gr_{\C}(r,D)$, or density cone\\
Hardware effects & None & None & Shots/noise if run on device & Shots/noise plus variational training\\
Gate complexity & Not applicable & Not applicable & Encoder dependent & Encoder and ansatz dependent\\
\bottomrule
\end{tabular}
\end{table}

\section{Resource accounting without hardware claims}
The trace identities have a practical resource implication, but this paper reports it only as accounting. No runtime decision rule, backend policy, or execution threshold is proposed. Table~\ref{tab:resource-accounting} summarizes which cost is being avoided or simulated in the controlled study.

\begin{table}[H]
\centering
\caption{Neutral resource accounting for the audited quantities. Here $N$ is the number of samples, $C$ the number of class or anomaly anchors, $D$ the Hilbert dimension, $r$ an anchor rank, and $S$ the number of shots. The table is descriptive and should not be read as a hardware-speedup claim.}
\label{tab:resource-accounting}
\footnotesize
\begin{tabularx}{\textwidth}{p{0.18\textwidth}p{0.23\textwidth}p{0.30\textwidth}X}
\toprule
Audited quantity & Circuit-style observation & Hermitian/geometric counterpart & Main interpretation\\
\midrule
Fidelity kernel & $O(N^2S)$ overlap observations & $O(N^2D)$ dense state inner products; closed form for separable maps & circuit estimate equals projector trace under ideal assumptions\\
Product-angle kernel & inverse-test or SWAP-test observations & $O(N^2q)$ closed form for separable $R_Y$ map & strong toy baseline can be classically cheap\\
Rank-$r$ return score & $O(NCS)$ success/failure observations & $O(NCDr)$ subspace overlaps after anchors are fixed & return probability equals anchor overlap\\
Density-center score & measurement or tomography-dependent estimate & dense density overlap, lower with factorized storage & mixed class models generalize pure anchors\\
Finite-shot audit & binomial observation noise & same trace score plus binomial sampling layer & shot uncertainty changes estimates, not the identity\\
\bottomrule
\end{tabularx}
\end{table}

\section{Why a reported QVR result could differ from a manifold result}
If a QVR implementation outperforms a simple manifold implementation, the trace identity is not contradicted. At least one modeling component differs.
\begin{enumerate}[leftmargin=2em]
\item \textbf{Different task.} Original QVR is one-class or anomaly-oriented; the present audit is supervised multiclass classification and regression.
\item \textbf{Different encoder.} Amplitude, phase, angle, time-evolution, and data-reuploading encoders define different state manifolds $\{P_x\}$.
\item \textbf{Different classifier.} A quantum-kernel SVM is a many-anchor margin model; pure QVR is a return-to-anchor model.
\item \textbf{Different observable.} Rank-one, rank-$r$, PSD, density, and general Hermitian observables are different geometric readouts.
\item \textbf{Different optimization budget.} A variational QVR unitary may train a circuit orbit, whereas an unrestricted Grassmann anchor is an upper-envelope model and a simple empirical anchor is a lower-cost baseline.
\item \textbf{Shots and noise.} Hardware QVR includes finite-shot statistics, transpilation, decoherence, readout error, and mitigation. A fair geometric counterpart must include the same observation model before hardware-level claims are compared.
\end{enumerate}

\section{Limitations}
The current manuscript has deliberately conservative boundaries.
\begin{enumerate}[leftmargin=2em]
\item The reported quantum experiments are statevector simulations, not hardware executions.
\item The multiclass QVR table is a supervised classwise adaptation of return scoring. A separate synthetic anomaly audit is included, but the paper still does not reproduce an external QVR time-series benchmark or hardware experiment.
\item Qubits are capped at six for the angle-circuit audit, causing PCA compression for high-dimensional tabular data.
\item The fixed entangling circuit is not trained.
\item The toy datasets are controlled witnesses, not real-world operational benchmarks.
\item Every table should be interpreted in light of the deterministic split descriptions, validation-selection discipline, and aggregate table summaries described in the manuscript.
\item Any table that displays best behavior across preprocessing or encoder variants must either be backed by validation-selection logs or described as post-hoc behavior analysis rather than a benchmark-selected estimate.
\end{enumerate}

\section{Conclusion}
The real-to-quantum geometry relationship is precise. The taxonomy developed here places common machine-learning and quantum-machine-learning scores on a common geometric axis while preserving their differences. A whole-span Grassmann projector removes ordered block allocation, while a flag projector retains it. A noiseless quantum fidelity kernel is a Hermitian projector kernel. A pure QVR return probability is a learned rank-one anchor score. A rank-$r$ return event is a complex Grassmann anchor score. A density-center score is a mixed-state extension. A quantum-kernel SVM is an affine Hermitian margin model rather than a single anchor.

The empirical lesson is not that quantum circuits or geometric lifts dominate classical baselines universally. They do not. The lesson is that normalization, projectivization, subspace projection, flag decomposition, density averaging, and circuit encoding are invariance assumptions. They help when those invariances match the task and fail, correctly, when the discarded information is label-bearing. The same-span flag witness and the native anomaly audit are two sides of the same message: representations should be chosen because their quotient geometry matches the scientific object. A useful external follow-up is a matched QVR--geometry study on public real time-series anomaly data, using the same encoder, return measurement, shot model, noise model, and training budget in both the circuit implementation and its Hermitian trace-score counterpart.

\section*{Statements and Declarations}
\textbf{Funding.} No external grant funding was received for this work. Pattern Recognition Pty Ltd provided internal financial and in-kind research support.

\textbf{Competing interests.} Azadeh Alavi is affiliated with RMIT University and Pattern Recognition Pty Ltd. Fatemeh Kouchmeshki and Hossein Akhoundi are affiliated with Pattern Recognition Pty Ltd. The authors and/or Pattern Recognition Pty Ltd have proprietary and commercial interests in geometry-aware quantum machine-learning software. A related Australian provisional patent application has been filed. The manuscript discloses mathematical identities, representation taxonomy, and public or synthetic-data experiments; proprietary implementation code and confidential operational materials are not disclosed.

\textbf{Data availability.} The empirical examples use public scikit-learn toy datasets and deterministic synthetic generators described in the manuscript. Aggregate table outputs and deterministic split descriptions may be made available from the corresponding authors upon reasonable request, subject to journal policy and removal of proprietary implementation materials. No private customer data or confidential company dataset is used in the manuscript.

\textbf{Code availability.} Executable implementation code is proprietary to Pattern Recognition Pty Ltd and is not publicly released. The manuscript provides the mathematical identities, modelling assumptions, synthetic-data descriptions, aggregate outputs, and audit definitions needed to assess the article-level claims. Non-proprietary table summaries may be made available from the corresponding authors upon reasonable request, subject to journal policy and removal of proprietary implementation materials.

\textbf{Ethics approval and consent to participate.} Not applicable. The study uses public toy datasets and synthetic data only.

\textbf{Author contributions.} All three authors contributed equally to the conception and design of the study, mathematical development, experimental design, interpretation of results, manuscript preparation, critical revision, and approval of the final manuscript. All three authors are corresponding authors.

\end{document}